\documentclass[11pt,letterpaper]{article}
\usepackage[margin=1in]{geometry}
\usepackage{microtype,caption,titlesec,placeins,needspace}
\titleformat{\section}{\large\bfseries}{\thesection}{.7em}{}
\titleformat{\subsection}{\normalsize\bfseries}{\thesubsection}{.7em}{}
\usepackage[T1]{fontenc}
\usepackage[utf8]{inputenc}
\usepackage{amsmath,amssymb,amsthm,graphicx,booktabs,url,algpseudocode,float}
\usepackage{xcolor}
\usepackage[colorlinks=true,linkcolor=black,citecolor=blue!50!black,urlcolor=blue!50!black]{hyperref}
\newtheorem{proposition}{Proposition}
\newcounter{paperalgorithm}
\newfloat{algorithmfloat}{tbp}{loa}
\newenvironment{paperalgorithm}[2]{\begin{algorithmfloat}[!htbp]\small\refstepcounter{paperalgorithm}\label{#2}\hrule\smallskip\textbf{Algorithm \thepaperalgorithm. #1}\par\smallskip}{\smallskip\hrule\end{algorithmfloat}}
\algrenewcommand\algorithmicrequire{\textbf{Input:}}
\algrenewcommand\algorithmicensure{\textbf{Output:}}
\newcommand{\RR}{\mathbb{R}}
\newcommand{\norm}[1]{\left\lVert#1\right\rVert}
\graphicspath{{figures/}}
\hypersetup{pdftitle={BRiDCT: Fast Two-Dimensional DCTs Using SIMD: SIMD Organization, Register Blocking, and Numerical Verification},pdfauthor={Antoine Moevus and Max Mignotte}}
\begin{document}
\raggedbottom
\widowpenalty=10000
\clubpenalty=10000
\begin{center}
{\Large BRiDCT: Fast Two-Dimensional DCTs Using SIMD\par}
\vspace{.4em}
{\large SIMD Organization, Register Blocking,\\and Numerical Verification\par}
\vspace{1.2em}
Antoine Moevus \quad Max Mignotte\\[.4em]
Department of Computer Science and Operations Research\\
Universit\'e de Montr\'eal\\[.7em]
\textbf{Technical Report}\\[.7em]
September 2026
\end{center}
\vspace{.7em}
\begin{abstract}
Many image-processing methods repeatedly compute two-dimensional discrete cosine transforms (DCTs). Fewer arithmetic operations do not necessarily make a DCT faster: execution time also depends on how calculations use the processor's single-instruction, multiple-data (SIMD) units, registers and memory. We implement and compare established DCT algorithms in SIMD to accelerate forward, inverse and round-trip transforms without approximating the transform matrix. The resulting C11 library, BRiDCT, combines the Shao--Johnson factorization with an organization of vector calculations and intermediate storage adapted to array size.

On one Apple M3 Max, BRiDCT is approximately 1.3--16.0 times faster than the tested general-purpose implementations---Apple Accelerate (vDSP), FFTW and Ooura---across their available single-precision (float32) cases from $8\times8$ to $256\times256$. It also outperforms the tested libjxl and floating-point libjpeg-turbo routines. Through its compiled NumPy interface, BRiDCT outperforms OpenCV in all tested cases, including squares up to $1024\times1024$ and two rectangular shapes, with Python-call and allocation costs included. Additional comparisons cover SciPy, DUCC and pyFFTW. These gains concern the tested interfaces and hardware; larger transforms called directly from C show size-dependent limits.

A publicly available verification dataset of 474 input arrays supports independent accuracy checks and exposes numerical limits at extreme amplitudes and for weak spectral components.\footnote{Code, versioned releases, dataset and recorded measurements: \url{https://github.com/antmoev/bridct}. See \hyperref[sec:availability]{Supplementary Material, Code and Data Availability} below.} The contribution is a faster implementation of the mathematical DCT through SIMD and memory organization, supported by measured performance and numerical verification.
\end{abstract}

\section{Introduction}
\subsection{The computational problem}
Many image-processing and numerical methods repeatedly apply a two-dimensional discrete cosine transform (DCT) and its inverse. We seek a fast float32 implementation for the six power-of-two square sizes from $8\times8$ to $256\times256$, with ordinary row-major input and output. The tasks are the forward DCT, its inverse, and their normalized round trip. A fast $8\times8$ routine alone does not solve this problem across all six sizes.

Fewer arithmetic operations do not guarantee shorter CPU time. SIMD units process several values together, independent instructions can overlap, and data access depends on registers and caches. Arithmetic savings can be offset by data rearrangement, longer dependency chains or temporary values stored and reloaded. A method with more operations may therefore run faster, and the preferred organization can change with size and architecture. FFTW documents this gap \cite{frigo2005design}. Our approach is to compare factorizations together with their SIMD and memory organization.

Image restoration and remote sensing illustrate this demand. Mignotte uses repeated $8\times8$ DCTs \cite{mignotte2010fusion}, while Touati, Mignotte, and Dahmane use $16\times16$ DCT descriptors for multimodal change detection \cite{touati2020multimodal}. Both use Ooura's C routines.

Our immediate motivation is phase unwrapping, where the local least-squares solver of Moevus and Mignotte uses repeated $16\times16$ forward and inverse transforms \cite{moevus2026tilebased}. BRiDCT also covers $8\times8$ blocks and larger arrays to serve the broader range of workloads above. We isolate the DCT from the surrounding reconstruction. Inputs may be whole images or application-supplied blocks.

We preserve the exact mathematical DCT. Here, \emph{exact} excludes coefficient truncation and deliberately approximate transform matrices. Haweel's signed DCT, for example, replaces cosine coefficients by their signs to avoid multiplications \cite{haweel2001new}. Exact does not mean bit-exact real arithmetic: floating-point errors and range limits are tested separately. We implement established algorithms in SIMD and select execution strategies by measured time subject to numerical verification. Section~\ref{sec:extensions} extends the primary size range to larger squares and rectangles.

\subsection{Existing algorithms and implementations}
FFTW generates computational kernels and selects execution plans \cite{frigo2005design}. Apple's Accelerate framework provides optimized DCT routines through vDSP.\footnote{Apple vDSP documentation: \url{https://developer.apple.com/documentation/accelerate/vdsp}.} Ooura's public C package includes general transforms and specialized $8\times8$ and $16\times16$ routines \cite{ooura2001general}. Image codecs offer further practical references: libjpeg-turbo supplies JPEG transform routines, and libjxl uses Highway SIMD. Their supported sizes, precision and layouts determine the comparable cases.

\paragraph{Arithmetic complexity and CPU time.}
Direct length-$n$ matrix evaluation uses $n^2$ multiplications and $n(n-1)$ additions. For power-of-two DCT-II sizes, Shao and Johnson reduce the established unitary count $2n\log_2 n-n+2$ to a scaled split-radix construction with leading term $(17/9)n\log_2 n+O(n)$ \cite{shao2008typeii}. Their tabulated totals decrease from 114 to 112 real additions/multiplications at $n=16$, and from 3842 to 3708 at $n=256$. These are scalar arithmetic counts, not CPU instructions. A separable square 2D transform applies $2n$ such one-dimensional transforms, with additional layout and normalization choices.

\paragraph{Code generation and evaluated implementations.}
SPIRAL automatically generates and searches vectorized transform implementations for a target architecture \cite{franchetti2002simd}. It establishes prior work on joint algorithm and SIMD optimization. BRiDCT uses our own Shao--Johnson generators and execution policies, not SPIRAL-generated code. We built the public SPIRAL generator, but the generation recipes tried did not yield a validated float32 NEON DCT for our protocol. SPIRAL is therefore absent from the timing comparison. Supplementary Section S1 documents this incomplete integration.

We implemented SIMD versions of Ooura and Shao--Johnson cores and compared alternative lane arrangements and generated arithmetic graphs. Our \emph{Ooura NEON} adaptation remains a separate small-block reference. Table~\ref{tab:coverage} distinguishes these study-specific implementations from public software.

\subsection{Approach and contribution}
BRiDCT (Banded and Register-blocked implementation of the DCT) combines four-row bands, padded storage and register blocking with Shao--Johnson. Small kernels retain simpler organizations when faster. Our contribution is the resulting implementation and an evaluation that connects complete source changes to measured gains, while checking numerical errors independently.

All primary results evaluate one fixed release on an Apple M3 Max. Figure~\ref{fig:organization} explains its SIMD organization. Section~\ref{sec:timing} evaluates native performance and Section~\ref{sec:publicresults} evaluates complete Python calls. Section~\ref{sec:validation} examines numerical limits. Earlier variants appear only in development ablations.

\section{Transform Contract and Algebraic Basis}
\subsection{Forward, inverse, and round trip}
For a side length $n$, define the orthonormal DCT-II matrix $Q_n\in\RR^{n\times n}$ by
\begin{equation}
 \begin{gathered}
 (Q_n)_{kj}=\alpha_k\cos\!\left[\frac{\pi}{n}\left(j+\frac12\right)k\right],\\
 \alpha_0=n^{-1/2},\qquad \alpha_k=(2/n)^{1/2}\quad(k>0).
 \end{gathered}
 \label{eq:q}
\end{equation}
For an input array $X$, the forward transform and its inverse are
\begin{equation}
 \mathcal C(X)=Q_nXQ_n^{\mathsf T},\qquad
 \mathcal C^{-1}(Y)=Q_n^{\mathsf T}YQ_n.
 \label{eq:2d}
\end{equation}
Equations~\eqref{eq:q}--\eqref{eq:2d} define the exact transform targeted by all kernels. Their factorizations and data permutations preserve this map in real arithmetic. Finite-precision outputs need not be identical across operation orders or fused multiply-add (FMA) settings. The inverse is the adjoint under the Frobenius inner product. We use \emph{inverse} throughout to distinguish this operator from a memory transpose, which merely exchanges array indices. The timed round trip is $\mathcal C^{-1}(\mathcal C(X))$: its exact result is $X$. It includes no nontrivial spectral filter and no Poisson solve.

Every comparison must deliver the same output normalization. An unnormalized DCT-II/DCT-III pair cannot be compared directly with \eqref{eq:2d} without accounting for its scale. For example, FFTW's REDFT10 and REDFT01 are unnormalized, with a pair factor $2n$ in one dimension \cite{frigo2005design}. The tested adapters absorb the necessary scale factors into their output or intermediate correction tables. Impulse responses check these factors before timing.

\subsection{Scaled cores and adjoints}
Write $Q_n=D_nG_n$, with scaled core $G_n$ and diagonal normalization $D_n$. For vectorized arrays, define
\begin{equation}
 H_n=G_n\otimes G_n,\qquad E_n=D_n\otimes D_n.
 \label{eq:scaled}
\end{equation}
The Kronecker products in \eqref{eq:scaled} give the 2D forward operator $U_n=E_nH_n$ and inverse $U_n^{\mathsf T}=H_n^{\mathsf T}E_n$.

\begin{proposition}[Folding diagonal scales]\label{prop:scale}
For any diagonal spectral operator $\Lambda$, the normalized transform--operator--inverse composition satisfies
\begin{equation}
 U_n^{\mathsf T}\Lambda U_n
 =H_n^{\mathsf T}(E_n\Lambda E_n)H_n.
 \label{eq:fold}
\end{equation}
In particular, the identity round trip requires only the diagonal correction $E_n^2$ between the unnormalized core and its adjoint.
\end{proposition}
\begin{proof}
Substitute $U_n=E_nH_n$ and $U_n^{\mathsf T}=H_n^{\mathsf T}E_n$. Associativity gives \eqref{eq:fold}; the middle factor is diagonal because all three factors are diagonal. With $\Lambda=I$, orthonormality gives $H_n^{\mathsf T}E_n^2H_n=I$.
\end{proof}
Proposition~\ref{prop:scale} lets forward scaling be applied in final stores and round-trip scaling through one correction table. Incorporating a solver's diagonal operator is possible algebraically, but no solver speedup is measured here.

\begin{proposition}[Adjoint of a generated linear graph]\label{prop:adjoint}
Consider a finite acyclic arithmetic graph whose nodes are additions and multiplications by constants. Reverse its data dependencies, replace each addition by accumulation into its input adjoints, and propagate through a constant multiplication using the same constant. The resulting graph evaluates the transpose of the original linear map in exact arithmetic.
\end{proposition}
\begin{proof}
For a local addition $z=a+b$, an output weight $\bar z$ contributes $\bar z a+\bar z b$ to the inner product, so it adds the same weight to both input adjoints. For $z=ca$, the contribution is $(c\bar z)a$. Applying these identities in reverse topological order preserves the inner product at every eliminated node. The final input weights therefore equal the matrix transpose applied to the initial output weights.
\end{proof}
The generator uses Proposition~\ref{prop:adjoint} for adjoint cores. Floating-point rounding, underflow and overflow still require tests.

\section{Implementation Organization}
\subsection{Two SIMD passes on a natural-layout array}
A separable 2D DCT applies one-dimensional transforms along both axes. The banded BRiDCT route uses four independent transforms per 128-bit NEON vector: four neighboring columns in the vertical pass, then four rows packed into a band in the horizontal pass. Figure~\ref{fig:organization} illustrates the change of lane assignment. The principal interface accepts and returns an ordinary row-major array.

Each four-lane vector processes four float32 values with one instruction, such as NEON's \texttt{vaddq\_f32} \cite{arm2026arm}. Data movement and dependencies prevent equating this with a guaranteed fourfold speedup.

\begin{figure}[!htbp]
\centering\includegraphics[width=.98\textwidth]{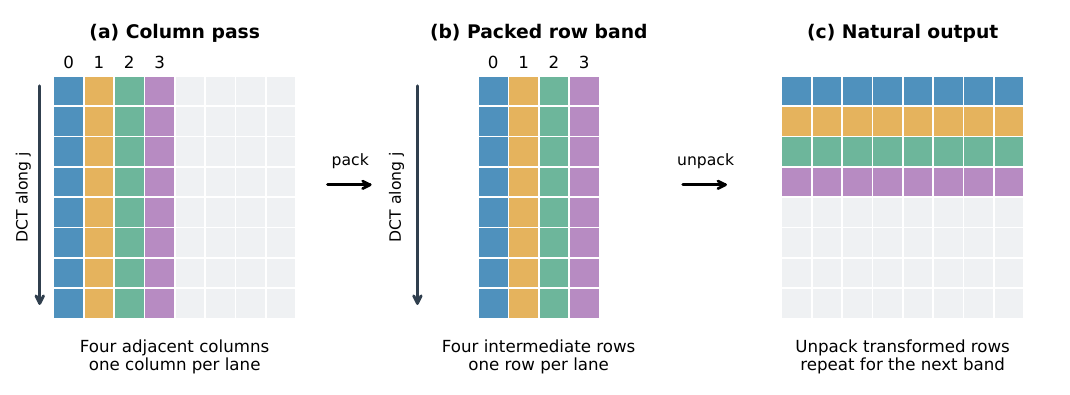}
\caption{Conceptual $8\times8$ example of the banded route; selected small kernels can use other organizations. Adjacent columns occupy the lanes in the vertical pass. Four intermediate rows are then packed as $T[j,\ell]=Z[r+\ell,j]$, transformed along $j$, and unpacked to row-major output. Colors identify lanes within each panel. The band occupies scratch memory; this is not a register-only transpose. Padding and scaling appear in Algorithm~\ref{alg:forward}.}\label{fig:organization}
\end{figure}

The vertical pass loads neighboring columns directly. The horizontal pass packs and unpacks four rows with interleaving loads/stores, including NEON \texttt{ld4}/\texttt{st4}-type operations. This removes a full-array transpose, but not all rearrangement in memory. Algorithm~\ref{alg:forward} gives the padded path. Smaller policies omit padding when it costs more than it saves.

\begin{paperalgorithm}{Natural-layout padded forward DCT}{alg:forward}
\begin{algorithmic}[1]
\Require Row-major $X\in\RR^{n\times n}$; preallocated scratch; scaled core $G_n$; diagonal factors $D_n$
\Ensure Row-major $Y=Q_nXQ_n^{\mathsf T}$ in exact arithmetic
\State Copy each input row into scratch with stride $n+4$
\For{each group of four columns}
 \State Apply $G_n$ down the columns, one column per SIMD lane
\EndFor
\For{each group of four intermediate rows}
 \State Pack the rows into a contiguous four-lane band
 \State Apply $G_n$ along that band, one row per SIMD lane
 \State Unpack to $Y$, applying the two-dimensional scale factors
\EndFor
\end{algorithmic}
\end{paperalgorithm}
Packing only permutes values, and each lane evaluates the same scalar transform on a different row or column. The two passes therefore implement \eqref{eq:2d}. The inverse uses the transposed core and corresponding input scales. The native round trip folds the intermediate correction using Proposition~\ref{prop:scale}.

\subsection{Blocking, padding, and specialization}
Expanded code exposes constants and avoids recursive calls, but large transforms can keep more values live than registers can hold. The resulting \emph{spills} store and reload temporary values. BRiDCT's blocked variant limits generated leaves to at most 16 points and expresses larger combinations with loops and coefficient tables organized by recursion level. This changes execution organization, not the underlying Shao--Johnson factorization.

Padded variants use a stride of $n+4$ floats. A power-of-two row stride can repeatedly visit a small subset of cache sets. Perturbing that stride can distribute accesses differently. This motivates the intervention, but timing alone does not identify the processor's physical cache-indexing function or isolate cache-conflict costs. Padding also adds scratch capacity and copies, so it is not a memory-footprint reduction.

The distributed policy fixes the variant and FMA setting for each size and direction before measurement. Generated kernels specialize loop bounds and coefficients, while every path retains the same DCT. Expanded, banded, padded and blocked forms remain useful at different sizes. Native calls reuse plan-owned tables and workspace, with no allocation or planning inside the transform. Python batches repeat these calls through a C loop. Supplementary Sections S1--S2 specify the workloads, and Section S5 evaluates the inverse-$16\times16$ refinement.

\section{Experimental Protocol}
\subsection{Native comparisons and provenance}\label{sec:nativeprotocol}
The primary timings use one Apple M3 Max. BRiDCT and its native harness are C11/NEON compiled with Apple Clang~21.0.0. Third-party binaries have their own recorded build configurations. Table~\ref{tab:coverage} distinguishes external software from adaptations made in this study. All comparisons require float32 orthonormal outputs. Input copies required by the measured adapters count toward runtime, while preparation and planning are excluded. All primary comparisons evaluate the final distributed BRiDCT configuration, version 0.2.0-dev.7, without selecting policies from these measurements. The native workload uses one worker and a rotating pool of four physical inputs, with one array per API call.

\begin{table}[!htbp]
\centering\small
\caption{Provenance and tested interfaces. Public software coverage refers to these routes, not every possible configuration.}\label{tab:coverage}
\begin{tabular}{p{.23\columnwidth}p{.67\columnwidth}}
\toprule
Implementation & Tested route\\\midrule
Apple vDSP & Accelerate; separable 1D calls and a 2D adapter, available sizes only.\\
FFTW & Public 2D real-to-real plans; adapter normalization; PATIENT or MEASURE as recorded.\\
Ooura & Public C routines; this study's NEON adaptation is identified separately.\\
libjxl & Public Highway SIMD route; one array per call, square sizes 8--256.\\
libjpeg-turbo & Public floating-point JPEG forward DCT; supported small-size case.\\
SciPy & Public \texttt{scipy.fft.dctn/idctn}; compiled PocketFFT backend.\\
DUCC & Public \texttt{ducc0.fft.dct}; type-II/type-III calls with orthonormal scaling.\\
OpenCV & Public \texttt{cv2.dct/idct}; one 2D image per call.\\
pyFFTW & FFTW through Python: cached builder and explicit real-to-real plan evaluated separately.\\
BRiDCT & Final Shao--Johnson C11 library; native plan API and compiled NumPy extension.\\
\bottomrule
\end{tabular}
\end{table}

FFTW uses true 2D plans. The vDSP and general Ooura adapters compose 1D calls and transposes. Integer JPEG transforms are excluded. Supplementary Section S1 specifies the sources and adapters.

Each of two fresh-process sessions uses 21 randomized blocks per arm with a 30\,ms calibration target. FFTW PATIENT planning is capped at one second per plan. Ratios and paired-bootstrap intervals use matched shapes and tasks within a session. All observations are retained, including blocks below the target. Known competing compiler and benchmark processes are checked, but exclusive desktop use is not guaranteed.

\subsection{Python calls and native execution}\label{sec:pythonprotocol}
The Python experiments measure the time an application pays for a complete call through the final compiled NumPy extension. They include dispatch, argument validation, the requested output allocation, native execution, normalization, copies and destruction of allocated outputs. A Python clock surrounds repeated calls. Elapsed time is divided by repetitions and the number of images processed. Inputs are contiguous float32 arrays. Imports, input generation, explicit planning and warmup precede timing. BRiDCT reuses its plan and workspace.

One comparison covers SciPy, DUCC and cached/planned pyFFTW. It uses ten shapes---eight square sides from 8 to 1024, plus $16\times32$ and $8\times64$---batches of one/four, and forward, inverse and materialized round trips, giving 60 cases per reference. All interfaces return fresh output. Each round trip uses separate forward and inverse calls and materializes the orthonormal spectrum. SciPy uses compiled PocketFFT \cite{virtanen2020scipy}; DUCC uses orthonormal type-II/type-III transforms; the two pyFFTW routes use cached convenience calls or reused 2D MEASURE plans, including timed copies and float32 normalization.

A separate OpenCV comparison covers 130 contracts over the same shapes: direct and batched calls, allocated and supplied outputs, and one- and two-entry round trips. BRiDCT loops over batches in C. OpenCV is called once per image per direction. These are practical interface comparisons with documented call boundaries, not identical wrappers or pure native-kernel timings. The one-entry BRiDCT round trip still materializes the spectrum. It is distinct from the fused native identity pair.

Each comparison uses two fresh-process sessions, 21 randomized paired blocks and a 30\,ms calibration target. Numerical validation precedes timing. One worker is requested. OpenCV's build-dependent thread query does not measure actual worker activity. Supplementary Section S2 records the packages, contracts and block durations. No estimated Python-call cost is subtracted. Native C and Python measurements remain separate, and absolute times from the two Python comparisons are not pooled.

\section{Timing Results}\label{sec:timing}
\subsection{Native performance and implementation choices}
Figure~\ref{fig:natural} and Table~\ref{tab:times} evaluate the final library at the six primary square sizes. BRiDCT is faster than FFTW PATIENT and general Ooura in all 18 size/task cases, and Apple vDSP in all 15 available cases, in each session. Ratios span $1.31$--$16.04$. It is also faster than the study-specific Ooura SIMD adaptation in 6/6 small-block cases in session~1 and 6/6 in session~2. It is also faster than its argument-check-free control in 6/6 small-block cases in session~1 and 6/6 in session~2.

\begin{figure}[!htbp]
\centering\includegraphics[width=.98\textwidth]{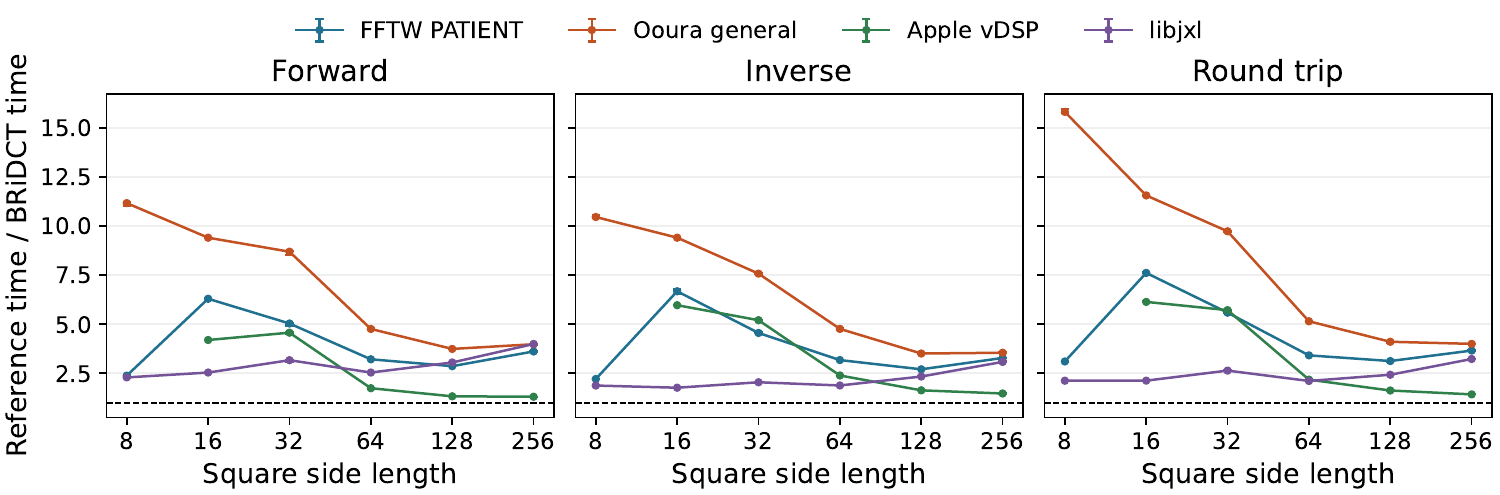}
\caption{Final BRiDCT configuration: session-2 reference/BRiDCT ratio of median native times; above one favors BRiDCT. The same fixed policy is used in both sessions. Bars are descriptive paired-bootstrap 95\% intervals over 21 blocks. One array is processed per call from a four-array pool, with natural layout and adapter copies included. vDSP has no size-8 measurement.}\label{fig:natural}
\end{figure}
\begin{table}[!htbp]
\centering\small
\caption{Final BRiDCT configuration: session-2 median native time in microseconds per natural-layout array. Bold marks the smallest displayed value in each row. F: forward; I: inverse; RT: identity round trip. Ooura SIMD is a study-specific adaptation. All entries in a row come from the same matched campaign. A dash denotes an unavailable route.}\label{tab:times}
\begin{tabular}{rlrrrrrr}
\toprule
Side & Task & BRiDCT & vDSP & FFTW & Ooura & Ooura SIMD & libjxl\\\midrule
8 & F & \textbf{0.022} & -- & 0.053 & 0.249 & 0.025 & 0.051\\
8 & I & \textbf{0.023} & -- & 0.052 & 0.245 & 0.024 & 0.044\\
8 & RT & \textbf{0.031} & -- & 0.097 & 0.493 & 0.035 & 0.066\\
16 & F & \textbf{0.091} & 0.382 & 0.573 & 0.856 & 0.097 & 0.231\\
16 & I & \textbf{0.091} & 0.544 & 0.608 & 0.857 & 0.101 & 0.161\\
16 & RT & \textbf{0.148} & 0.906 & 1.124 & 1.707 & 0.164 & 0.314\\
32 & F & \textbf{0.432} & 1.976 & 2.178 & 3.760 & -- & 1.370\\
32 & I & \textbf{0.488} & 2.541 & 2.223 & 3.699 & -- & 0.997\\
32 & RT & \textbf{0.762} & 4.359 & 4.260 & 7.419 & -- & 2.010\\
64 & F & \textbf{2.771} & 4.827 & 8.912 & 13.200 & -- & 7.043\\
64 & I & \textbf{2.765} & 6.609 & 8.781 & 13.178 & -- & 5.201\\
64 & RT & \textbf{5.056} & 11.011 & 17.271 & 26.055 & -- & 10.673\\
128 & F & \textbf{14.821} & 19.685 & 42.467 & 55.499 & -- & 45.252\\
128 & I & \textbf{15.518} & 25.346 & 41.974 & 54.470 & -- & 36.264\\
128 & RT & \textbf{26.163} & 42.504 & 81.814 & 107.490 & -- & 63.529\\
256 & F & \textbf{62.248} & 81.326 & 225.031 & 247.937 & -- & 248.633\\
256 & I & \textbf{68.768} & 101.242 & 225.758 & 244.086 & -- & 212.059\\
256 & RT & \textbf{120.090} & 171.441 & 439.680 & 480.719 & -- & 387.984\\
\bottomrule
\end{tabular}

\end{table}
Supplementary Figure S1 includes the public floating-point codec routes and the study-specific small-block Ooura controls, measured against the same final BRiDCT build.

Figure~\ref{fig:interventions} reports the development ablation that motivated the final design, rather than another ranking of the final library. It follows source changes at fixed FMA policy in the original matched campaign. At 128, replacing full-array transposes with bands reduces the forward time from $41.69$ to $22.59\,\mu$s; padding gives $18.61\,\mu$s, and blocked recursion $15.15\,\mu$s. At 256 the sequence is $165.34$, $104.37$, $83.36$, and $63.13\,\mu$s. Each effect is conditional on the preceding variant: the ratios are not independent cache, permutation or spill costs. Smaller arrays can favor simpler organizations. Assembly inspection finds fewer static spill/reload sites in blocked functions, but supplies no dynamic traffic count. Combining the two one-dimensional passes into a single arithmetic graph reduced operations but did not reduce runtime in the tested variant. Execution organization and arithmetic count must therefore be evaluated together.

\begin{figure}[!htbp]
\centering\includegraphics[width=.98\textwidth]{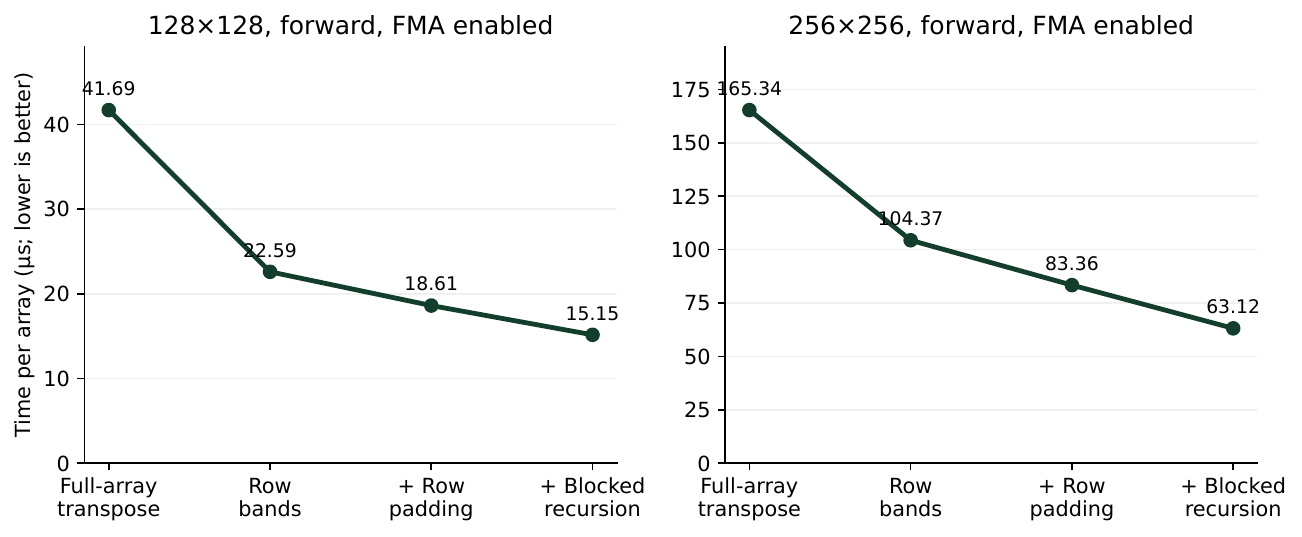}
\caption{Forward-transform interventions at fixed FMA policy and natural layout: band buffers, padded stride, then blocked recursion. Differences measure each complete source change, not isolated cache or spill costs. These historical variants explain the design choices; final-library comparisons are measured separately in Figure~\ref{fig:natural}.}\label{fig:interventions}
\end{figure}

\subsection{Complete-call performance of Python interfaces}\label{sec:publicresults}
The compiled NumPy interface exposes reusable plans, allocated or supplied outputs and a C loop over batches. Figure~\ref{fig:publicpython} shows complete-call performance under Section~\ref{sec:pythonprotocol}'s contracts. BRiDCT has lower medians in all 60 tested cases per interface against SciPy, DUCC, cached pyFFTW and planned pyFFTW, in both sessions. The separate OpenCV campaign favors BRiDCT in all 130 call contracts in each session. The minimum OpenCV/BRiDCT ratio over all contracts in both sessions is $3.90$. These complete-call ratios include the prescribed output ownership and batch handling.

\begin{figure}[!htbp]
\centering\includegraphics[width=.98\textwidth]{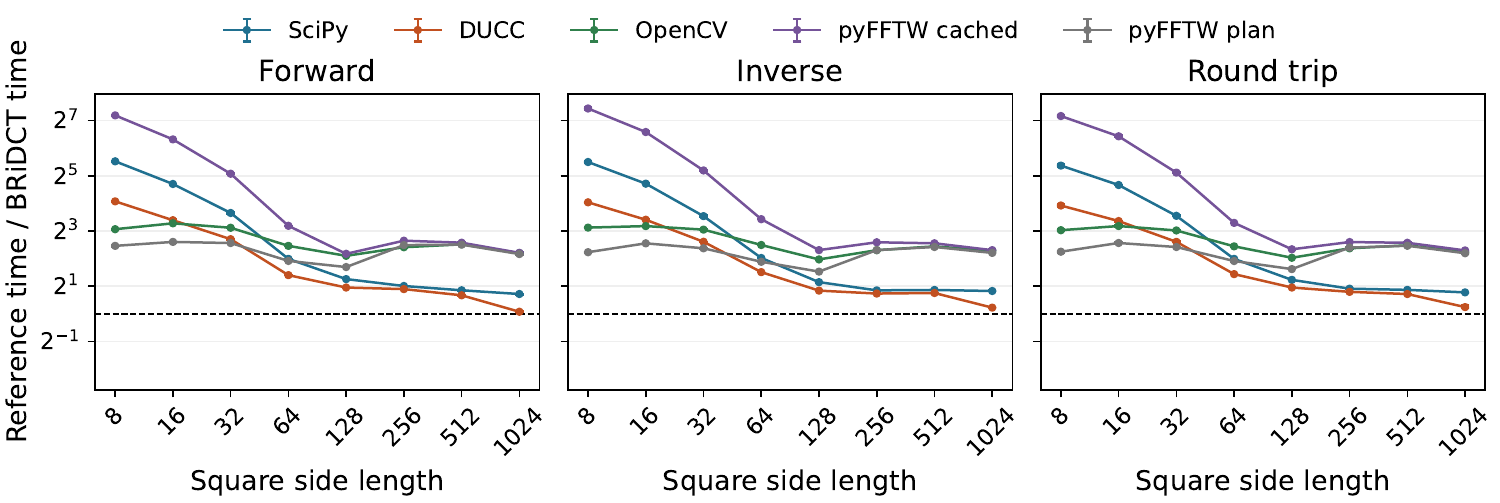}
\caption{Final compiled NumPy interface: complete Python-call reference/BRiDCT ratio, one image per call, session~2. The round trip uses separate forward and inverse calls for both methods. Each curve uses paired measurements within its own comparison; OpenCV and the other libraries are measured in separate sessions, so absolute times are not pooled. Bars are descriptive paired-bootstrap 95\% intervals; call, allocation, copies and normalization are included. Cached and planned pyFFTW are distinct interfaces to FFTW.}\label{fig:publicpython}
\end{figure}
\begin{samepage}
These complete-call gains include interface and allocation behavior. They do not rank native kernels: C BRiDCT and Python OpenCV have different call boundaries. Supplementary Section S2 provides the tables.
\par\end{samepage}

\section{Numerical Verification}\label{sec:validation}
A forward and inverse implementation can share an error and still pass a round trip. We therefore compare all three tasks independently with a float64 cosine-matrix oracle. A deterministic dataset of 474 float32 arrays contains 420 core cases and 54 boundary diagnostics across the six primary sizes. Core inputs include impulses, cosine modes, cancellation, neighboring floats and weak perturbations of a strong constant, with sampled amplitudes from $2^{-100}$ to $2^{90}$. Boundary cases probe subnormals and extreme finite amplitudes. Figure~\ref{fig:corpus} illustrates these inputs. Supplementary Section S4 records the data and scripts.

\begin{figure}[!htbp]
\centering\includegraphics[width=.98\textwidth]{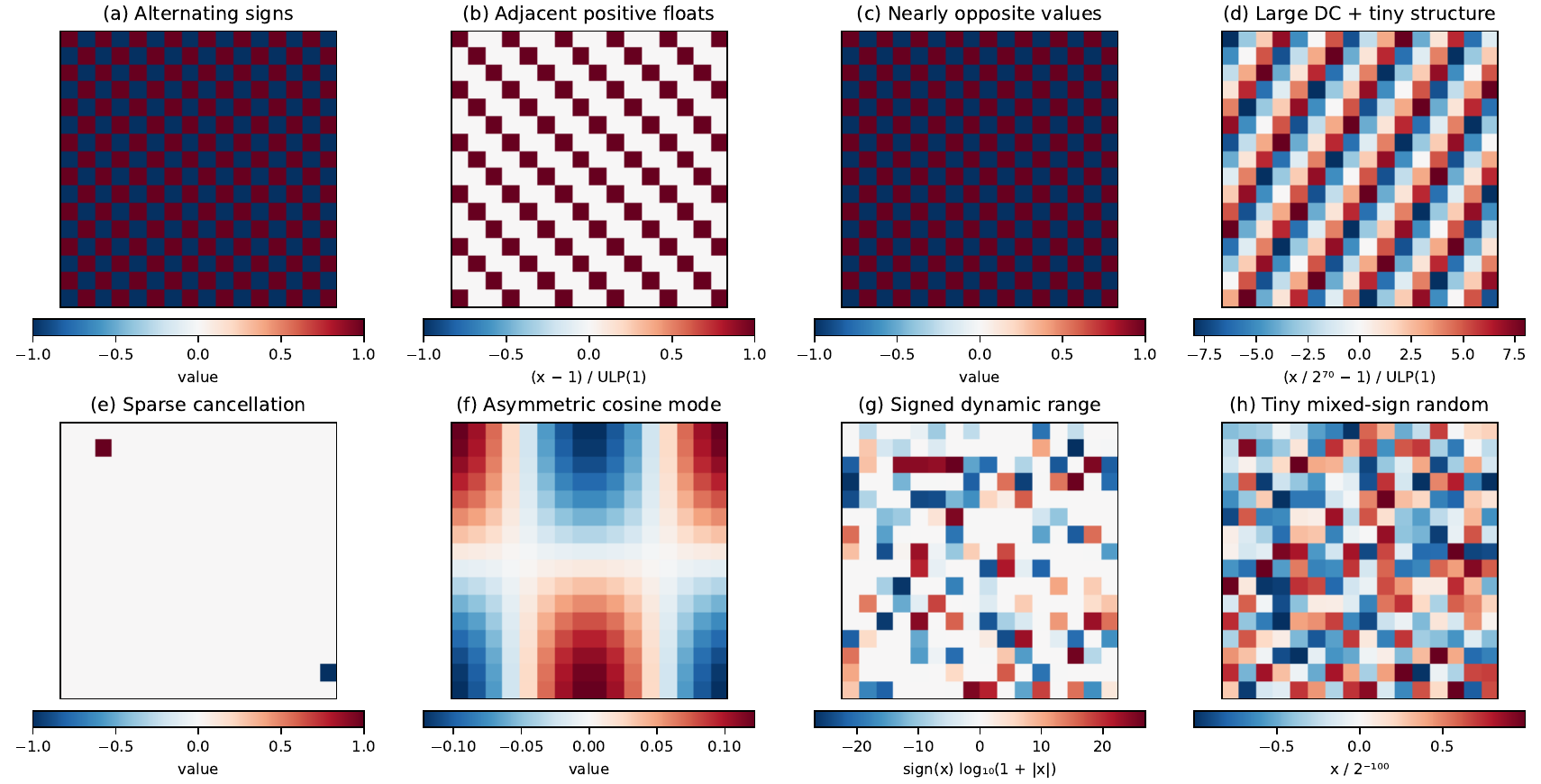}
\caption{Representative $16\times16$ verification inputs. ULP is the spacing to the next float32 value. Display rescalings, stated below each panel, are not applied during testing. Exact stored values distinguish nearly cancelling neighbors and weak components that may look alike in an image.}\label{fig:corpus}
\end{figure}

For input $X$, reference output $Y$ and computed output $\widehat Y$, the core criterion is
\begin{equation}
 e_{\mathrm{rel}}=\frac{\norm{\widehat Y-Y}_F}{\norm{X}_F}< 2\times10^{-5}.
 \label{eq:error}
\end{equation}
The Frobenius norm is the square root of the sum of squared entries. Zero inputs use an absolute check. Orthonormality makes \eqref{eq:error} an output-relative norm criterion as well. No denominator floor conceals errors on tiny inputs. This is an empirical acceptance threshold, not an analytical bound, and the float64 oracle has its own rounding error.

Each core check tests one input and one transform direction with the final library's fixed policies. All 1\,260 core checks pass, with maximum discrepancy $2.251\times10^{-7}$ (Table~\ref{tab:numerical}). Inputs and guard values surrounding buffers remain unchanged. Separate procedural tests of the 64 supported shapes pass 1\,152 required C checks. Together with the dataset this gives 2\,412 required checks, with maximum discrepancy $2.325\times10^{-7}$. The compiled Python interface independently passes 1\,152 cases with allocated and supplied outputs.

Injected scaling, sign, orientation and component-loss faults are detected in all 78 combinations of size, direction and fault on at least one input. This tests the criterion's sensitivity, not universal defect detection.

\begin{table}[!htbp]
\centering\small
\caption{Final-library numerical audit by array side length. Maximum error refers to the core inputs and the input-relative Frobenius criterion in \eqref{eq:error}. The last column reports failed boundary diagnostics over all boundary checks; these extreme-range cases are not included in the core accuracy claim.}\label{tab:numerical}
\begin{tabular}{rrrr}
\toprule
$N$ & Core checks & Maximum $e_{\mathrm{rel}}$ & Diagnostic failures \\
\midrule
8 & 210 & $1.37\times10^{-7}$ & 12/27 \\
16 & 210 & $1.39\times10^{-7}$ & 14/27 \\
32 & 210 & $1.60\times10^{-7}$ & 16/27 \\
64 & 210 & $1.98\times10^{-7}$ & 16/27 \\
128 & 210 & $1.92\times10^{-7}$ & 16/27 \\
256 & 210 & $2.25\times10^{-7}$ & 16/27 \\
\bottomrule
\end{tabular}

\end{table}

\begin{samepage}
Boundary diagnostics show 90 failures among 162 checks: unnormalized intermediate stages can overflow even when the final orthonormal answer is representable, while subnormals can suffer large relative underflow/rounding errors. A strong constant can also hide weak-component loss: a weak-perturbation case has about 8.95\% relative discrepancy in its nonconstant coefficients despite passing the global norm criterion. Applications depending on weak spectral detail therefore need component-specific checks. Buffer guards and this finite dataset do not establish universal numerical or memory safety.
\par\end{samepage}

\section{Discussion and Limits}
\subsection{Larger squares, rectangles, and unsupported dimensions}\label{sec:extensions}
For $X\in\RR^{h\times w}$, the rectangular operator is $Q_hXQ_w^{\mathsf T}$. The final library supports independent power-of-two axis lengths from 8 through 1024. The two lengths need not match. The native campaign includes larger squares, reversed rectangles, $16\times32$ and $8\times64$, and explicit availability probes with an axis of length 26. All 2\,192 direction checks of the available native routes pass before timing; these include the references as well as BRiDCT. The final library's broader numerical checks are reported in Section~\ref{sec:validation}.

The native protocol of Section~\ref{sec:nativeprotocol} applies unchanged. Supplementary Section S3 gives the times and ratios. At $512\times512$, BRiDCT is faster than the tested vDSP route in all three tasks in session~2. At $1024\times1024$, BRiDCT is faster than vDSP only for the inverse. Orientation matters: the $8\times64$ and $64\times8$ BRiDCT round trips take $0.597$ and $0.834\,\mu$s, respectively. vDSP also leads for the $512\times1024$ forward transform. These observations do not establish a ranking for every rectangle.

Non-power-of-two axes remain unsupported: BRiDCT rejects $32\times26$ and $26\times32$, whereas FFTW computes those transforms. This is a compatibility limit, not a numerical failure. Zero-padding would change the DCT and is not used as a substitute.

\subsection{Scope of the measured gains}
The measured gains concern the specified hardware, interfaces and warmed workloads. They do not establish a universal fastest-DCT ranking. Input layout, first-use costs and the fraction of application time spent in transforms can change the practical benefit. Energy consumption and dynamic memory traffic are unmeasured.

\section{Future Work}
Multiplatform work will tune register blocking and memory layout while retaining Shao--Johnson and the common numerical contract. Priorities include a controlled final-release Ubuntu performance evaluation, large-format Windows performance, Linux ARM64 and additional compilers, including MSVC. Supplementary Section S5 records current checks and the compiler, header and fallback-kernel issues encountered. Continuous integration (CI) checks builds and accuracy on macOS, Ubuntu and Windows; performance evaluation requires controlled-machine measurements.

Further comparisons should cover Intel IPP, alternative Accelerate routes, FFTW/SPIRAL-generated kernels and MATLAB under matched precision. MATLAB's documented double output requires accounting for precision and MEX call costs.\footnote{MATLAB API: \url{https://www.mathworks.com/help/images/ref/dct2.html}.} Interface controls, hardware counters and first-use/streaming workloads can separate call, allocation and memory effects. Large and rectangular transforms, exact non-power-of-two support and numerical-range protection remain optimization targets. Full applications and GPU implementations \cite{obukhov2008discrete} require end-to-end measurements, including spectral operations, data transfers and scheduling.

\section{Conclusion}
BRiDCT delivers fast, numerically verified two-dimensional DCTs by combining the Shao--Johnson factorization with SIMD lane organization, banded intermediate storage, padding and register blocking. It demonstrates how adapting execution to the processor can accelerate the exact mathematical transform beyond what arithmetic counts alone predict.

On the tested Apple M3 Max, BRiDCT is approximately $1.3$--$16.0\times$ faster than the tested general-purpose DCT implementations---FFTW, general Ooura and Apple's Accelerate vDSP route---across their available forward, inverse and round-trip cases from $8\times8$ to $256\times256$. It also outperforms the tested public floating-point codec routes at their available sizes. The native advantage over these general-purpose routes extends to the three tested $512\times512$ tasks in both sessions. Its compiled NumPy interface outperforms OpenCV in every tested call contract, including squares up to $1024\times1024$ and two rectangles, with gains of at least $3.90\times$ including call and allocation costs. Native $1024\times1024$ transforms remain a target for further optimization.

\section*{Supplementary Material, Code and Data Availability}\phantomsection\label{sec:availability}
The accompanying \emph{Supplementary Methods and Results} details the protocols, numerical diagnostics, extended-shape comparisons and portability checks.

The source code, its version history, the 474-input verification dataset and the recorded benchmark measurements are publicly available on GitHub: \url{https://github.com/antmoev/bridct}. Versioned releases provide source packages with build and test scripts for macOS, Ubuntu and Windows, as well as separate dataset and measurement archives. \href{https://github.com/antmoev/bridct/releases/tag/v0.2.0-dev.7-r1}{Release v0.2.0-dev.7-r1} preserves the numerical implementation and input arrays evaluated in this report. Original code is distributed under BSD-3-Clause and the dataset under CC BY 4.0; third-party components retain their licenses and attribution.

\end{document}